\documentclass[twoside]{article}
\usepackage[a4paper]{geometry}
\usepackage[latin1]{inputenc} 
\usepackage[T1]{fontenc} 
\usepackage{RR}
\usepackage{hyperref}
\usepackage{cite}

\usepackage{amsmath,amssymb,mathtools}
\usepackage{algorithmic}
\usepackage{amsthm}

\usepackage{array}
\newtheorem{definition}{\textbf{Definition}}[section]
\newtheorem{property}{\textbf{Property}}[section]

\newtheorem{theorem}{\textbf{Theorem}}[section]

\usepackage{url}
\usepackage{dsfont}
\usepackage{bbold}
\RRNo{8846}
\RRdate{Janvier 2016}
\RRauthor{
Jean-Marie Gorce\thanks{JM Gorce is with Universit\'e de Lyon, INSA Lyon, Inria, CITI, 
F-69621 Villeurbanne, France}%
  \and
H. Vincent Poor\thanks{H.V. Poor is with Electrical Engineering Dept, Princeton University, Princeton, NJ-08544, USA}%
\and Jean-Marc Kelif \thanks{J-M Kelif is with Orange Labs. Issy Les Moulineaux, France} }
\authorhead{Gorce \& Poor \& Kelif}
\RRtitle{Extension des mod\`eles de canaux gaussiens asym\'etriques (canal \`a acc\`es multiple et canal \`a  diffusion) pour la prise en compte d'un continuum spatial d'utilisateurs.}
\RRetitle{Spatial Continuum Extensions of Asymmetric Gaussian Channels (Multiple Access and Broadcast)}
\titlehead{SCBC and SCMAC models}
\RRnote{Contactl: jean-marie.gorce@insa-lyon.fr}
\RRnote{This work has been submitted for publication to IEEE International Symposium on Information Theory (ISIT), 2016. }
\RRnote{This work has been done in the framework of the Inria associated team Cowin. The work of JM Gorce was supported by Orange Labs under Grant $n^oF05151$. The work of H.V. Poor was supported by the U.S. National Science Foundation under Grant ECCS-1343210. }
\RRresume{Ce rapport pr\'esente un nouveau mod\`ele appel\'e \emph{canaux asym\'etriques pour continuum spatial} destin\'e \`a l'\'etude la r\'egion de capacit\'e de canal pour des sc\'enarios asym\'etriques. Ces sc\'enarios asym\'etriques correspondent soit au cas o\`u soit une source transmet des flux d'information priv\'es vers une densit\'e spatiale d'utilisateurs ou au cas o\`u une densit\'e spatiale d'\'emetteurs transmettent des flux d'information priv\'es vers un unique r\'ecepteur.
Cette approche exploite les classique canal \`a diffusion (broadcast channel, BC) et canal \`a acc\`es multiple (multiple access channel, MAC). Cette \'etude, pour des raisons de place, est limit\'ee \`a l'\'etude de canaux Gaussiens avec contraintes de puissance et est restreinte au r\'egime asymptotique (capacit\'e sans erreur).

Le sc\'enario de r\'ef\'erence est constitu\'e d'une station de base (BS) en mode \'emission (Tx) ou r\'eception (Rx), d'une distribution spatiale al\'eatoire de noeuds (resp. en mode Rx ou Tx) caract\'eris\'ee par une densit\'e spatiale de probabilit\'e $u(x)$, auxquels est associ\'ee une quantit\'e d'information requise sans contrainte de d\'elai. 
Ce syst\`eme est mod\'elis\'e par un canal asym\'etrique avec un nombre infini d'utilisateurs en mode BC ou MAC. 
Pour \'etablir les propri\'et\'es de ce mod\`ele, une discr\'etisation spatiale multi-r\'esolution est effectu\'ee, permettant d'affiner autant que n\'ecessaire la repr\'esentation discr\`ete du probl\`eme.
Des r\'esultats d'atteignabilit\'e et de capacit\'e  sont obtenus en limite de la s\'erie constitu\'ee par cette approche multi-r\'esolution.

La capacit\'e uniforme est alors d\'efinie comme le d\'ebit sym\'etrique maximal atteignable par les utilisateurs en mode \'emission ou r\'eception sans contrainte de d\'elai.
La r\'egion de capacit\'e est \'egalement d\'efinie comme l'ensemble des distributions d'information atteignables.

La pr\'ecision de ces propri\'et\'es aux limites et leur int\'er\^et pratique sont bri\`evement discut\'ees.}

\RRabstract{This paper proposes a new model called \emph{spatial continuum asymmetric channels} to study the channel capacity region of asymmetric scenarios in which either one source transmits to a spatial density of receivers or a density of transmitters transmit to a unique receiver.
This approach is built upon the classical broadcast channel  (BC) and multiple access channel (MAC). For the sake of consistency, the study is limited to Gaussian channels with power constraints and is restricted to the asymptotic regime (zero-error capacity).

The reference scenario comprises one base station (BS) in Tx or Rx mode, a spatial random distribution of nodes (resp. in Rx or Tx mode) characterized by a probability spatial density $u(x)$ and a request for a quantity of information with no delay constraint. This system is modeled as an $\infty-$user  asymmetric channel (BC or MAC). 
To derive the properties of this model, a spatial discretization is performed and the equivalence with either a BC or MAC is established. A discretization sequence is then defined to refine infinitely the approximation. Achievability and capacity results are obtained in the limit of this sequence. The uniform capacity is then defined as the maximal symmetric achievable rate at which the distributed users can transmit/receive with no delay constraint.
The capacity region is also established as the set of information distributions that are achievable. 

The tightness of these limits and their practical interest are briefly illustrated and discussed. }
\RRmotcle{r\'eseaux cellulaires, canal \`a diffusion, canal \`a acc`\`es multiple, continuum spatial, capacit\'e, th\'eorie de l'information en r\'eseau}
\RRkeyword{Cellular networks, Broadcast channel, Multiple access channel, Spatial continuum channels, capacity, network information theory}
\RRprojets{Socrate}
\URRhoneAlpes 
\RCGrenoble 

\begin{document}
\makeRR   
\newpage
\tableofcontents
\newpage

\section{Introduction}
Asymmetric spatial continuum channels refer to scenarios in which either a unique transmitter sends information to a spatial density of receivers (i.e., a spatial continuum BC, abbreviated  SCBC), or in which a spatial density of transmitters send independent information streams to a unique receiver (i.e., a spatial continuum MAC, or SCMAC). The unique transmitter or receiver is called the base station (BS) and the distributed users are characterized with a spatial density function $ u(x)$.

Asymmetric dense wireless networks correspond to classical scenarios in modern communication systems. Basically, cellular systems alternate downlink and uplink transmissions which respectively correspond to BC and MAC. 
Information theory provides exact expression for the capacity regions for the classical BC and MAC under memoryless stationary Gaussian  assumptions \cite{el2011network}, for a predetermined set of users with fixed channel pdfs. But the capacity region of these  asymmetric wireless networks where the users are distributed according to a probability distribution has not been determined. 
Stochastic geometry allows a step forward by providing an estimate of the SINR distribution \cite{andrews2011tractable,elsawy2013stochastic} in cellular networks with randomly placed users. However, to the best of our knowledge, all derivations of cell rates from these distributions have been based on pure time sharing strategies, thus underestimating the capacity region of the cell.
In \cite{sang2014load}, the cell load is computed from an approximation of the cells' size distribution. But this work estimates the requested sum-rate per cell (the throughput demand), and not  the capacity of the radio access. 
In \cite{bonald2003wireless} and more recently in \cite{minelli2014optimal}, efficient cell capacity metrics are defined. But these metrics do not correspond to the Shannon capacity of the corresponding BC but rather to a rate achievable with some time-sharing. The gap between these results and a fundamental limit is not known.
In \cite{gorce2014energy} the fundamental energy efficiency - spectral efficiency tradeoff (EE-SE) in a dense typical cell was evaluated, which may be interpreted as a capacity although this was not explicitly proved. Only some achievable EE-SE tradeoffs were provided. 
The work  proposed herein may be seen as an extension of this former work,  using an information theoretic formalism. In this setting, the capacity regions of the SCBC and the SCMAC are defined and determined.

The main results of this paper are the following;
\begin{itemize}
\item The uniform capacities of the SCBC and SCMAC are defined and computed for Gaussian stationary memoryless channels.
\item The access capacity regions of the SCBC and SCMAC are assessed under the same assumptions.
\item The tightness of this model for a simple scenario is illustrated.
\end{itemize}

\section{Model and notations}
Although the approach can be more general, the problem studied in this paper is restricted to Gaussian channels where the BS and users are equipped with single antennas.
The maximum rate simultaneously achievable by all users (the symmetric rate) \cite{liang2006cth13} is first investigated. To avoid confusion between  symmetry/asymmetry of channels and rates, we will rather refer to this assumption as the {\it uniform rates} assumption. The uniform capacity is then defined and computed.
This result is extended to compute the access capacity region, defined as the set of achievable rate densities. 

The first asymmetric channel studied bellow is the SCBC. Then, the MAC/BC equivalence stated in \cite{jindal2004duality} is used to establish the SCMAC capacity when power is assumed  to be transferable between users. 

Consider a unique BS serving a large area with a high number of users. The area covered by the BS is denoted by $\Omega \subset \mathbb{R}^2$.  We denote by $(\Omega, \mathcal{A}, m)$  the corresponding measurable space with $ \mathcal{A}$ the Lebesgue $\sigma-$ algebra and $m$ the Lebesgue measure.
Let $x$ be a point in $\Omega$.

\begin{figure}[!t]
\centering
\includegraphics[width=4.5in]{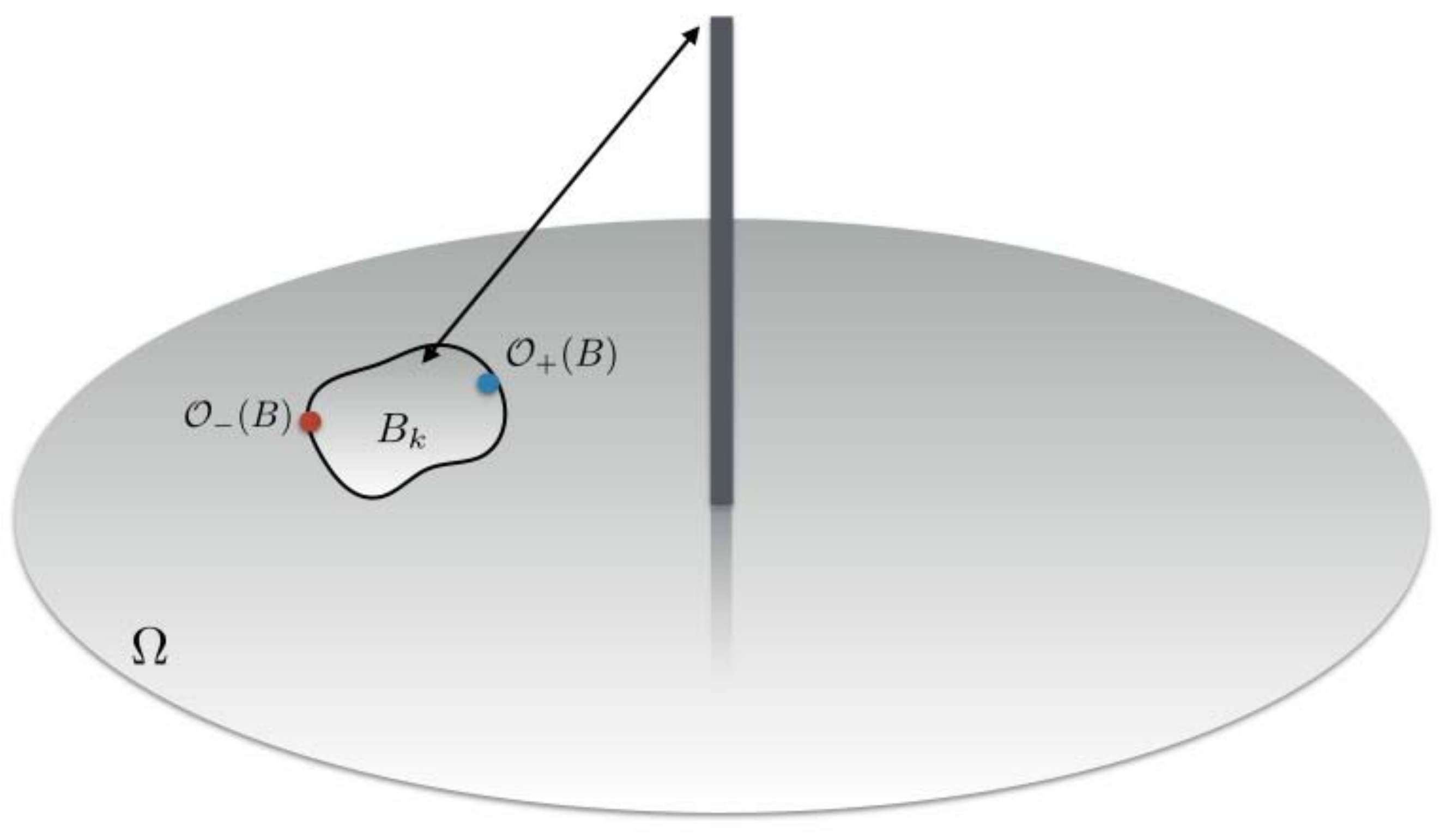}%
\hfil
\caption{Network model and feasible node. A physically feasible node (Def.2.2) can associated with any subset $B\in \mathcal{A}$. Its requested rate $\mathcal{R}(B)$ is obtained by the integration of the rate density over $B$ (Prop.2.1). Worst and best receivers, resp. $\mathcal{0}_-$ and $\mathcal{0}_+$ are two points of $B$.}
\label{fig:first}
\end{figure}
Without lack of generality, the BS is  assumed to be located at point $(0,0)$. Users appear randomly in time and space on $\Omega$. As such, they are not described by a discrete set but through a probability function $u(x)$ representing the probability density that a node appears at $x$. Denoting by $\tilde{u}(x,t)$ a realization of the ergodic random process with probability density function (pdf) $u(x)$, then, for any subset $B\in\mathcal{A}$, the average number of users is given by 
\begin{equation}
U(B)= \lim_{T\rightarrow\infty}\frac{1}{T}\int_{t\in T} \int_{x\in B} \tilde{u}(x,t) \cdot dx \cdot dt .
\end{equation}
The global average number of users associated with the whole space $\Omega$ is denoted by $U_T$.

\begin{definition}[Requested rate density]
The requested rate density $\rho(x): \Omega \rightarrow \mathbb{R}$ is a measurable function  that represents the information rate density requested at point $x$. 
\end{definition}
\begin{itemize}
\item Note 1: In the definition above, $\rho(x)$ is normalized by the system bandwidth and is expressed in $bps$ per $Hz$ per $m^{2}$. 
\item Note 2: This definition is valid as well for SCBC and SCMAC scenarios. The requested rate may represent either an uplink or a downlink stream.
\end{itemize}
When a uniform rate per user is considered,  each user requests the same quantity  of information denoted by $\mathcal{I}_0$. 
Therefore, the requested rate density is proportional to the users' density:
\begin{equation}
\rho(x)=\mathcal{I}_0 \cdot u(x) .
\end{equation}


\begin{property}
\label{prop:Irate}
The rate requested by a subset $B\subset\Omega$ is given by:
\begin{equation}
\label{eq:IB}
\mathcal{R}(B)\leq\int_B \rho(x)\cdot m(dx) .
\end{equation}
with equality if all requested information streams are independent.
\end{property} 
Assuming independence for the SCBC is equivalent to considering BC with private messages only while for the SCMAC, independence means that the sources are not correlated. Independence will be assumed throughout the rest of this paper.

The difficulty of this approach is to give a physical meaning of a spatial density of information . Indeed, what does a rate density represent physically? Our approach follows three steps: (i) We first partition the spatial continuum into subsets, each of them containing a reference node, called a {\it physically feasible node} (see Fig.\ref{fig:first}); 
(ii) we then introduce a splitting process to iteratively refine the discretization providing a sequence of partitions
(iii) Finally, we use the limit of this sequence to define and compute an asymptotic capacity region. 

In this context, we define:
  \begin{definition}[Physically feasible node]
\label{def:physically}
For any subset $B\in \mathcal{A}$,  a physically feasible node is a virtual node (either in Tx or Rx mode) that requests the quantity of information  $\mathcal{R}(B)$ given by \eqref{eq:IB}.
\end{definition}

Then, let  $\mathcal{B}=\left\{B_k;k\in [1;K]\right\}$ be a partition of $\Omega$ with $B_k \in \mathcal{A} ; \forall k$. According to Def.\ref{def:physically}, a physically feasible node $v_k$ is associated with each element $B_k$.
Therefore, the BS and the set of nodes $\{v_1,\dots,v_K\}$ form either a $K-$user BC or a $K-$user MAC. This system, denoted by $\mathfrak{N}(\mathcal{B})$, is called a physically feasible network.

\begin{definition}[Sequence of physically feasible networks]
\label{def:seqPFN}
Consider a sequence of partitions $\mathcal{B}^{(i)}$ for $i\in\mathbb{N}$ with  $\mathcal{B}^{(0)}=\left\{B_0^{(0)}=\Omega\right\}$ and where a splitting process divides each set $B^{(i)}_{k}$ into two subsets $\left\{B^{(i+1)}_{2k},B^{(i+1)}_{2k+1}\right\}$.  
A physically feasible node $v^{(i)}_k$ is associated with each $B^{(i)}_{k}$.
The sequence of networks $\left\{\mathfrak{N}(\mathcal{B}^{(i)})\right\}$ is called a sequence of {\it physically feasible networks} and is denoted by $\mathfrak{N}^{(i)}$
\end{definition}
When $i\rightarrow\infty$, the size of $\mathfrak{N}^{(i)}$ grows to infinity while the individual requested rates tend to $0$. In the meantime the requested sum-rate remains constant. 

\section{Gaussian-SCBC: definition and properties}
The previous formalism is used to derive the following results for the Gaussian-SCBC:
\begin{itemize}
\item The minimal power needed to serve a given requested rate density is established. 
\item As a corollary, the uniform capacity of the SCBC is given.
\item The access capacity region of the SCBC is also formulated as the set of feasible rate densities. 
\end{itemize}
The Gaussian-SCBC characterizes the outputs associated to the users distribution over $\Omega$ as a function of an input symbol $y$ at the BS.

Let $\Xi$ be the set of piecewise integrable continuous functions $\Xi = \left\{ \Phi(x): \Omega \rightarrow \mathbb{R}^d\right\}$ where $\Phi(x)$ is an observable field on $\Omega$. $d$ is the dimension of the field. For instance, $d=1$ for a real channel while for a complex baseband signal, $d=2$. $d$ could be even larger, e.g. for multi-antennas or multiband receivers.  

\begin{definition}[Gaussian-SCBC]
\label{def:G-SCBC}
Given the following:
\begin{itemize}
\item $\Omega$, a subspace on a Hilbert space of dimension 2,
\item $x_0=(0,0) \in \Omega$, the source,
\item$\mathcal{Y}_c$, the coding alphabet used by the source to transmit a symbol $y \in \mathcal{Y}_c$,
\item $\Xi = \left\{ \Phi(x): \Omega \rightarrow \mathbb{R}^d\right\}$ a set of fields on $\Omega$,
\end{itemize}
The SCBC is a function that maps any input code $y$  to a set of conditional pdfs on $\Xi$:
\begin{equation}
\mathcal{H} :  \left\{\mathcal{P}_{\Phi(.)|y} ; \forall y\in \mathcal{Y}_c\right\},
\end{equation}
\end{definition}
which extends the classical BC definition to include a spatial continuum of potential receivers.

The set of conditional pdfs associated with the Gaussian-SCBC is
\begin{equation}
\label{eq:AWGN}
\mathcal{P}_{\Phi(x)|y} = \mathcal{N}\left( y ,\frac{\sigma^2}{ h(x)} \right),
\end{equation}
where $\mathcal{N}(\mu,v)$ denotes for the normal distribution of mean $\mu$ and variance $v$.  In \eqref{eq:AWGN} the channel gain is normalized and the noise variance is proportional to the inverse of the channel gain coefficient $h(x)$. 

\subsection{Relative transmission technique}
A transmission technique is now defined with respect to (w.r.t.) a partition $\mathcal{B}$ of $\Omega$.
For each element $B_k$, the former definition of a physically feasible node is extended to that of a physically feasible receiver. 
\begin{definition}[Physically feasible receiver]
A {\it physically feasible receiver} associated with an element $B_k$ of a partition $\mathcal{B}$ is defined by two successive operations (observation and decoding) applied to the field $\Phi(x);x\in B_k$ :
\begin{equation}
\Phi_{x\in B_k} \overset{\mathcal{O}}{\longrightarrow} V_k \overset{\mathcal{D}}{\longrightarrow}   \hat{m}_k .
\end{equation}
\end{definition}
The observation operator $\mathcal{O}$ plays a fundamental role in  the proposed analytical approach: it extracts some information from the local field $\Phi(x); x\in B_k$. Then, the decoder conventionally maps this observation to an estimate $\hat{m}_k$.  While the observer is imposed as a part of the transmission system, the decoder is usually chosen to optimize the transmission.  

Let us now assume that a physically feasible receiver is only able to observe a unique point $x$ on $B_k$. Two observers are defined:  the best observer $\mathcal{O}_{+}$ which selects the least noisy sample over $B_k$, and the worst observer $\mathcal{O}_{-}$ which selects the most noisy sample.  When one of these observation techniques is associated with a partition $\mathcal{B}$, one obtains resp. the best and the worst physically feasible networks, denoted by $\mathfrak{N}_{+}(\mathcal{B})$ and $\mathfrak{N}_{-}(\mathcal{B})$.

\begin{definition}[Relative transmission technique]
A transmission technique $(M_1,\dots,M_K,n,\epsilon)$ associated with a network $\mathfrak{N}(\mathcal{B})$ is given by a joint information message $M$ (built from the $K$ individual messages) in $n$ channel uses and where each physically feasible receiver observes and decodes an estimate of the message $M_k$, denoted by $\hat{m}_k$ with an average error probability lower than $\epsilon$. 
A transmission technique is asymptotically feasible without error if $\lim_{n\rightarrow \infty}\epsilon = 0$.

\end{definition}
Under these assumptions and for a given observer, $\mathfrak{N}(\mathcal{B})$  is equivalent to a classical Gaussian-BC for which the capacity region is  known \cite{cover1972broadcast}.  

\subsection{Asymptotic achievability and converse}
First, let us note that the term \emph{asymptotic} does not refer as above to the coding length $n$ but rather to the partition index $i$ in the sequence $\mathfrak{N}^{(i)}$.
Let first the achievability be expressed w.r.t. a network $\mathfrak{N}(\mathcal{B})$:
\begin{definition}[Relative achievability]
\label{def:relative_achievability}
The requested rate density $\rho(x)$ is said to be achievable w.r.t. to a network  $\mathfrak{N}(\mathcal{B})$ if a transmission technique $(M_1,\dots,M_K,n,\epsilon)$ exists such that $|M_k|\geq 2^{n\cdot\mathcal{R}(B_k)}, \forall k$.
\end{definition}

\begin{definition}[Asymptotic achievability]
\label{def:asymptachievability}
Consider a sequence of physically feasible networks  $\mathfrak{N}^{(i)}$ (see Def.\ref{def:seqPFN}).
The requested  rate density$\rho(x)$ is said to be {\it asymptotically achievable} if $\rho(x)$ is achievable w.r.t. $\mathfrak{N}^{(i)}$ when $i\rightarrow\infty$.
The requested rate density $\rho(x)$ is said to be {\it doubly asymptotically achievable} if the transmission technique further satisfies $\lim_{n\rightarrow \infty}\epsilon = 0$.
\end{definition}
In the following, only  doubly asymptotic achievability is studied. 
Thus achievability is implicitly studied under the asymptotic regime.

\begin{definition}[Access Capacity region]
\label{def:accesscapacity}
The access capacity region of a Gaussian-SCBC, denoted by $\mathcal{U}_\Omega$,  is the set of doubly asymptotically achievable rate densities $\rho(x)$.
\end{definition}

Based on these definitions, we can state the two following theorems.
\begin{theorem}[Relative achievability with worst observers implies asymptotic achievability]
\label{theo:achievability}
Consider a sequence of physically feasible networks using the worst receivers $\mathfrak{N}_{-}^{(i)}$. 
If $\rho(x)$ is achievable w.r.t. $\mathfrak{N}_{-}^{(i)}$ for some $i\geq 0$, then $\rho(x)\in \mathcal{U}_\Omega$.
\end{theorem}
\begin{proof}
The proof relies on the properties of the degraded BC. 
The network $\mathfrak{N}_{-}^{(i)}$ forms a $K$-user BC with joint requested rates  $\left(\mathcal{R}_k^{(i)}=\mathcal{R}\left(\mathcal{B}_k^{(i)}\right); \forall k \right)$.
$\rho(x)$ is achievable w.r.t. $\mathfrak{N}_{-}^{(i)}$ if these rates belong to its capacity region.
It is enough to prove that achievability w.r.t. $\mathfrak{N}_{-}^{(i)}$, implies achievability w.r.t. $\mathfrak{N}_{-}^{(j)}; \forall j>i$.

Let assume that the following transmission technique exists:  
\begin{equation*}
T^{(i)}=(M^{(i)}_1,\dots,M^{(i)}_K,n\rightarrow \infty,\epsilon\rightarrow 0).
\end{equation*}
We have to prove that $\exists T^{(i)} \Rightarrow \exists T^{(j)}; \forall j>i$.
Let us prove that when one set $B_k^{(i)}$ is split, the requested rates of its children $B_{2k}^{(i+1)}$ and  $B_{2k+1}^{(i+1)}$ are achievable simultaneously and jointly with the rates of the other sets.
Let the message $M_k$ be built by combining two messages:
\begin{equation}
\label{eq:subcoding}
\left(M^{(i+1)}_{2k},M^{(i+1)}_{2k+1}\right) \rightarrow M^{(i)}_{k}
\end{equation}
By definition, one of the two children has an observer identical to the parent's one (noted $V_{2k}^{(i+1)}$), while the other observation $V_{2k+1}^{(i+1)}$ is better than or equal to this one. 
These two observations and $Y$ form a  Markov chain \cite{el2011network}:
\begin{equation*}
Y\longrightarrow V_{2k+1}^{(i+1)} \longrightarrow V_{2k}^{(i+1)}.
\end{equation*} 

From Prop.\ref{prop:Irate} :
\begin{equation*}
\mathcal{R}_k^{(i)}=\mathcal{R}_{2k}^{(i+1)}+\mathcal{R}_{2k+1}^{(i+1)}.
\end{equation*}
Since $V_{2k}^{(i+1)}=V_{k}^{(i)}$, $\mathcal{R}_k^{(i)}$ is achievable from $V_{2k}^{(i+1)}$ and  by the data processing inequality,  also from $V_{2k+1}^{(i+1)}$.\\
Therefore, a simple time sharing in the implementation of \eqref{eq:subcoding} ensures simultaneously the desired rates.
Doing similarly for each $B_k$ proves that the relative achievability with the worst receiver is preserved with the splitting process and thus ensures asymptotic achievability.
\end{proof}

\begin{theorem}[Converse : relative non-achievability with best observers implies asymptotic non-achievability]
\label{theo:converse}
Consider a sequence of physically feasible networks (see Def.\ref{def:seqPFN}) using the best receivers $\mathfrak{N}_{+}^{(i)}$. 
If a requested rate density $\rho(x)$ is proved to be not achievable w.r.t. $\mathfrak{N}_{+}^{(i)}$ for some $i\geq 0$, then $\rho(x)$ is not asymptotically achievable.
\end{theorem}
\begin{proof}
The network $\mathfrak{N}_{+}^{(i)}$ forms a $K$-user BC with joint requested rates  $\left(\mathcal{R}_k^{(i)}=\mathcal{R}\left(\mathcal{B}_k^{(i)}\right); \forall k \right)$.
$\rho(x)$ is not achievable w.r.t. $\mathfrak{N}_{+}^{(i)}$ if these rates do not belong to its capacity region.
This means that any transmission technique $T^{(i)}=(M^{(i)}_1,\dots,M^{(i)}_K,n,\epsilon)$ fails to achieve at least one requested rate $\mathcal{R}_k^{(i)}$.
To prove the theorem it is enough to prove that the capacity region shrinks with the splitting process, or equivalently to prove that achievability w.r.t. $\mathfrak{N}_{+}^{(i)}$ implies achievability w.r.t. $\mathfrak{N}_{+}^{(j)}; \forall j<i$.
Assuming that the rates $\mathcal{R}_{2k}^{(i)}$ and $\mathcal{R}_{2k+1}^{(i)}$ are simultaneously achievable, then it is obvious that the sum-rate is also achievable by the best receiver according to the following Markov chain relationship:
\begin{equation*}
\left(M^{(i)}_{2k},M^{(i)}_{2k+1}\right) \rightarrow M^{(i-1)}_{k} \rightarrow V_{2k}^{(i)} \rightarrow V_{2k+1}^{(i)} .
\end{equation*}
Again, the data processing inequality implies that the link $M^{(i-1)}_{k}  \rightarrow V_{2k}^{(i)}$ has a higher capacity than the end-to-end link,
which proves that the splitting process with the best receiver reduces achievability.
\end{proof}

To illustrate these two theorems, consider the initial partition $\mathcal{B}^{(0)}=\left\{\Omega\right\}$. 
The best and worst observers correspond respectively to the best and the worst users over $\Omega$. The theorems above state that $\rho(x)$ is achievable if its sum-rate $\mathcal{R}(\Omega) \leq \mathcal{C}(\mathfrak{N}_{-}^{(0)})$ and cannot be achievable if $\mathcal{R}(\Omega) \geq \mathcal{C}(\mathfrak{N}_{+}^{(0)})$.
These bounds  are obvious but not tight. 

However, when $i\rightarrow \infty$, the capacity regions with the worst and the best receivers may converge asymptotically to the same region, which is then defined as the capacity region of the Gaussian-SCBC. The convergence is granted if the splitting process is done such that the difference between the best and the worst observations tends to $0$ for each $B_k^{(i)}$. 

\subsection{Uniform capacity}
The uniform capacity of the SCBC is defined as follows.
\begin{definition}[Uniform capacity]
Consider a density of users $u(x)$, and a BS with a maximal transmission power $P_M$. The uniform capacity is defined by $\mathcal{C}_0 = \sup_{\rho(x)\in \mathcal{U}_\Omega} ( \mathcal{I}_0 )$.
\end{definition}
As a corollary, we can define the minimal power required to serve a given $\mathcal{I}_0$:
\begin{definition}[Minimal power]
Consider a distribution of users $u(x)$, with an average information request $\mathcal{I}_0$. The minimal power $\tilde{P}_m$ required to ensure the asymptotic achievability is defined by $\tilde{P}_m = \min_{P\in \mathbb{R}} \left[ P ; \rho(x)\in \mathcal{U}_\Omega(P)\right] $.
\end{definition}

These definitions are duals, but it  more convenient to compute $\tilde{P}_m$ as a function of $\mathcal{I}_0$.
Let be defined the complementary cumulated distribution function (ccdf) of the users' noise variance as follows:
\begin{equation}
\label{eq:ccdf}
G_\nu(\nu)=\frac{1}{U_T} \cdot \int_\Omega u(x)\cdot\mathbb{1}\left[\nu(x)\geq \nu\right] \cdot dx.
\end{equation}
This function is related to the noise pdf as follows:
\begin{equation}
f_\nu(\nu)=-\frac{\delta}{\delta \nu}G_\nu(\nu). 
\end{equation}

Finally, define the Gaussian SCBC spectral efficiency $\eta_s=\log(2)\cdot \mathcal{I}_0 \cdot  U_T$, which represents the system efficiency in nats per channel use. 

\begin{theorem}[Gaussian-SCBC Minimal power]
\label{theo:SCBC-unifcapa}
The minimal power required to serve a given  user density $u(x)$ and a quantity of information $\mathcal{I}_0$ is given by
\begin{equation}
\tilde{P}_m = 2 \eta_s \int_{\nu_m}^{\nu_M} x\cdot f_\nu(x) \cdot e^{2 \eta_s \cdot  G_\nu(x)} \cdot dx 
\end{equation}
\end{theorem}
\begin{proof}
The proof exploits theorems \ref{theo:achievability} and \ref{theo:converse} with a specific sequence $\mathfrak{N}^{(i)}$. 
Assume, without loss of generality, that $\forall x; \nu(x) \in [\nu_m;\nu_M)$.
The SNR at any receiver is given by $\gamma(x)=\frac{\tilde{P}_m}{\nu(x)}$ and the corresponding capacity is $\frac{1}{2}\log(1+\gamma(x))$.

Considering $B_0^{(0)}=\Omega$, $\mathcal{I}_0$ is bounded by the two theorems \ref{theo:achievability} and \ref{theo:converse} :
\begin{equation}
\frac{\log(1+\gamma_m)}{2 U_T}\leq \mathcal{I}_0 \leq \frac{\log(1+\gamma_M)}{2 U_T }
\end{equation}

A sequence of partitions of $\Omega$ is now built via a splitting process based on the values of $\nu(x)$. 
Assuming that for any subset $B_k^{(i)}$ , the noise values belong to some interval $[\nu_{m,k}^{(i)},\nu_{M,k}^{(i)})$, then $B_k^{(i)}$ is split such that
\begin{align*}
& B_{2k}^{(i+1)}=\left\{ x\in B_k^{(i)}; \nu(x) <  \bar{\nu}_{k}^{(i)} \right\} \\
& B_{2k+1}^{(i+1)}=\left\{ x\in B_k^{(i)}; \nu(x) \geq \bar{\nu}_{k}^{(i)} \right\}.
\end{align*}
where  $\bar{\nu}_{k}^{(i)}$ is a threshold value, which can be chosen as $\bar{\nu}_{k}^{(i)}= (\nu_{m,k}^{(i)}+\nu_{M,k}^{(i)})/2 $.
This process is equivalent to a progressively refined discretization of the users' equivalent noise distribution as illustrated in Fig.\ref{fig:second}.
\begin{figure}[!t]
\centering
\includegraphics[width=4.5in]{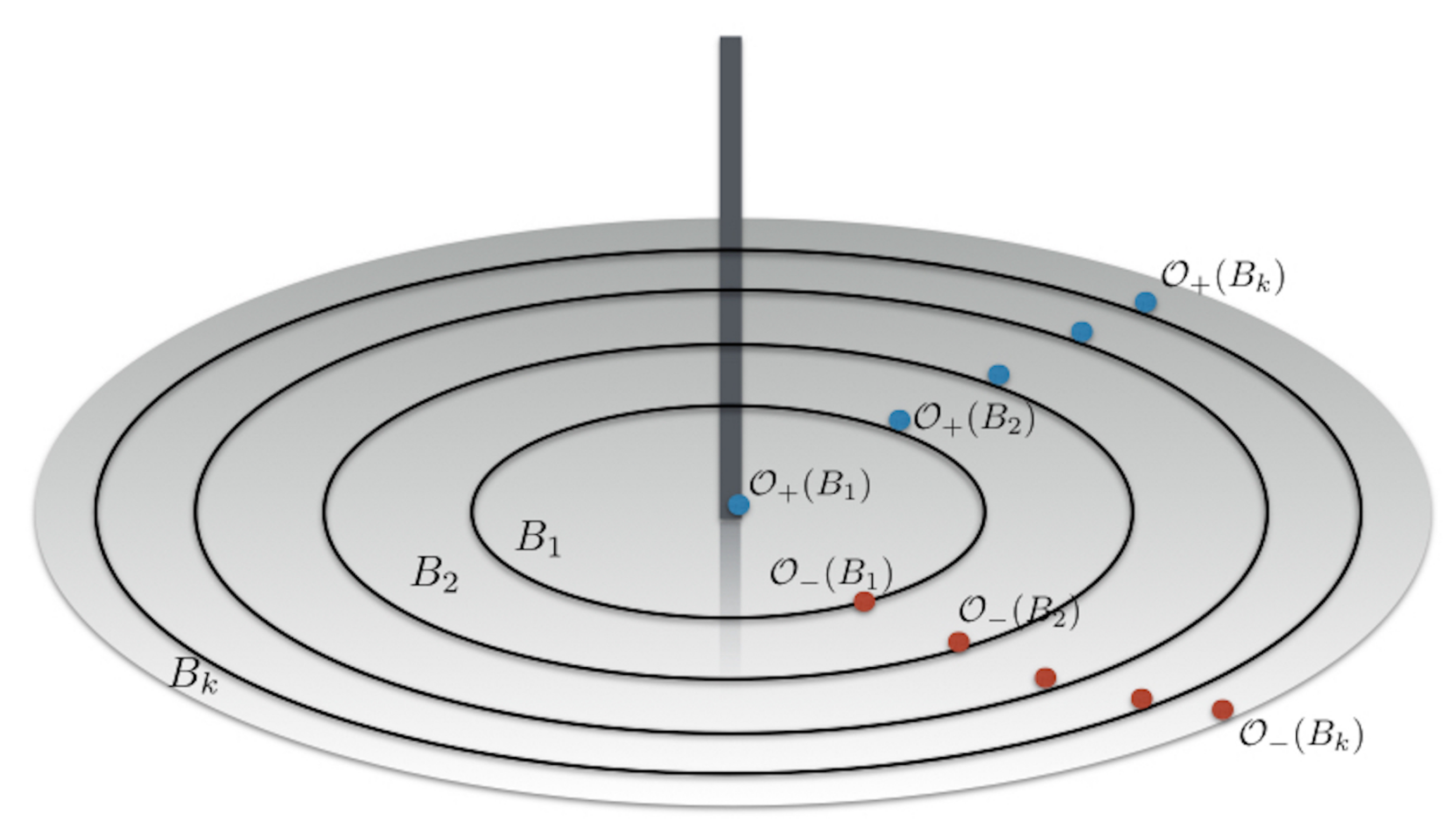}%
\hfil
\caption{Illustration of the partition of $\Omega$ in homogeneous regions, with their best (blue) and worst (red) receivers. Note that in this Figure, the best receiver associated to each subset is picked up among the nearest points to the BS (inward circle) while the worst receiver is picked up among the furthest points (outward circle). }
\label{fig:second}
\end{figure}

The requested rate associated with a subset $B_k^{(i)}$ is
\begin{equation}
\mathcal{R}_k^{(i)}=\mathcal{I}_0\cdot u_0\cdot \int_{\nu_{m,k}^{(i)}}^{\nu_{M,k}^{(i)}} f_\nu(\nu) \cdot d\nu.
\end{equation}

According to  theorems \ref{theo:achievability} and \ref{theo:converse}, at level $i$, the uniform capacity $\mathcal{C}_0$ is upper and lower bounded by the capacity of the physically feasible networks $\mathfrak{N}_{+}$ and $\mathfrak{N}_{-}$, which correspond to $2^{i}$-user BCs.
The capacity of a $K$-user BC is known to be achieved with  superposition coding by ranking the receivers from the worst to the best user. 
Due to the splitting process defined above, the threshold noise values are ordered and satisfy
\begin{align*}
& \nu_{m,0}^{(i)} < \nu_{m,1}^{(i)} < \dots < \nu_{m,2^{i}-1}^{(i)}  \\
& \nu_{M,k}^{(i)} = \nu_{m,k+1}^{(i)}
\end{align*}
Therefore it is possible to bound $\tilde{P}_m$ using the superposition coding technique with either the best or the worst receivers.
The superposition coding principle implies use of the following power for each block $k$:
\begin{align*}
& \mathfrak{N}_{+} : P_{+,k}^{(i)}= \left( 2^{2\mathcal{R}_k^{(i)}} - 1 \right)\cdot \left(\nu_{m,k}^{(i)} + \sum_{q < k} P_{+,q}^{(i)} \right)   \\
& \mathfrak{N}_{-} : P_{-,k}^{(i)}= \left( 2^{2\mathcal{R}_k^{(i)}} - 1 \right)\cdot \left(\nu_{M,k}^{(i)}  + \sum_{q< k} P_{-,q}^{(i)} \right)
\end{align*}
where $ P_{.,k}^{(i)}$ is the power density required to serve the $k^{th}$ user of level $i$ with either the best or the worst receiver.

The accumulated power required to serve all subsets up to $k$ is denoted by $ \Pi_{\cdot,k}^{(i)}$ and is given by the recursive sum of the previous equations:
\begin{equation*}
\mathfrak{N}_{\cdot} :  \Pi_{\cdot,l}^{(i)}= \sum_{k\leq l}\left( 2^{2\mathcal{R}_k^{(i)}} - 1 \right)\cdot \left(\nu_{\cdot,k}^{(i)}+ \Pi_{\cdot,k-1}^{(i)} \right) .
\end{equation*}
We denote by $L(i)$ the number of blocks at level $i$. Then, we have
\begin{equation}
 \Pi_{+,L(i)}^{(i)} \leq \tilde{P}_m \leq  \Pi_{-,L(i)}^{(i)}
\end{equation}
When $i\rightarrow \infty$, $\nu_M - \nu_m \rightarrow 0$, then $ \Pi_{+,k}^{(i)} -  \Pi_{-,k}^{(i)} \rightarrow 0$. 
And since  $\lim_{i\rightarrow\infty} \left( 2^{2\mathcal{R}_k^{(i)}} - 1\right) = 2\log2\cdot\mathcal{R}_k^{(i)}$, $\tilde{P}_m$ becomes the solution of the following Riemann integral:
\begin{equation}
 \Pi(\nu)= 2\eta_s \cdot  \int_{\nu_m}^{\nu}  f_\nu(x) \cdot (x + \Pi(x) ) \cdot dx ,
\end{equation}
for $\nu=\nu_M$. Then, writing the derivative
\begin{equation}
 \dot{\Pi}(\nu)=2 \eta_s  f_\nu(\nu) \cdot (\nu + \Pi(\nu) ),
\end{equation}
 and solving it  leads to theorem \ref{theo:SCBC-unifcapa}.
\end{proof}

\subsection{Access capacity region}
The result above can be exploited to estimate the access capacity region as defined in Def.\ref{def:accesscapacity}.

Assume that the network wants to transmit some information density directly characterized by a density $\rho(x)=\rho_T\cdot f_\rho(x)$where  $f_\rho(x)$ and $\rho_T$ are respectively the normalized traffic distribution and the sum-rate.  The equivalent noise distribution can be  obtained as:
\begin{equation}
f_\nu(\nu)= \frac{\partial}{\partial \nu}\left[ \int_\Omega f_\rho(x)\cdot\mathbb{1}\left[\nu(x)<\nu\right] \cdot dx \right], 
\end{equation}

Extending theorem \ref{theo:SCBC-unifcapa}, a distribution $\rho(x)$ is asymptotically achievable iff $\rho_T$ and $f_\nu(\nu)$ jointly satisfy
\begin{equation}
\label{eq:rho0}
 \rho_T \int_{\nu_m}^{\nu_M} x\cdot f_\nu(x) \cdot e^{2\log(2) \rho_T \cdot  G_\nu(x)} \cdot dx \leq P_ t  / (2\log(2)) .
\end{equation}
Because the left-side expression is strictly positive and monotonically increasing with $\rho_T$, for any normalized traffic distribution $f_\rho(x)$, there exists a maximal value of $\rho_T$ under which asymptotic achievability is  satisfied, which characterizes the access capacity region.

\section{Gaussian-SCMAC: definition and properties}

The SCMAC  considered in this section ais defined as the dual of the SCBC:
\begin{definition}[Gaussian-SCMAC]
\label{def:G-SCMAC}
Given the following:
\begin{itemize}
\item $\Omega$, a subspace on a Hilbert space of dimension 2,
\item$\mathcal{X}_c$, a coding alphabet,
\item $\Xi = \left\{ \Phi(x): \Omega \rightarrow   \mathcal{X}_c \right\}$, a set of input fields on $\Omega$,
\item $x_0=(0,0) \in \Omega$, the receiver position,
\item $y \in \mathbb{R}^d$  the channel output,
\end{itemize}
The SCMAC is a function that maps any input field $\Phi(x)$  to a set of conditional pdfs on $\mathbb{R}^d$:
\begin{equation}
\mathcal{H} :  \left\{\mathcal{P}_{y|\Phi(.)} ; \forall \Phi(.) \in \Xi \right\}.
\end{equation}
\end{definition}

The set of conditional pdfs associated with the Gaussian-SCMAC is
\begin{equation}
\label{eq:AWGN2}
\mathcal{P}_{y| \Phi(x)} = \mathcal{N}\left( \int_\Omega h(x)\cdot \Phi(x) \cdot dx ,\sigma^2 \right),
\end{equation}
where $h(.)$ represents a linear channel. 

\subsection{Relative transmission technique}
As for the SCBC, let us now define a transmission technique w.r.t. a partition of $\Omega$.

\begin{definition}[Physically feasible transmitter]
A {\it physically feasible transmitter} associated with an element $B_k$ of a partition $\mathcal{B}$ is defined by a selector  $\mathcal{S}$  that selects a transmission point $x_k\in B_k$, and a coding technique that maps a message to a coding value $M_k \rightarrow \Phi_{ B_k}(x) = y_k \cdot \delta(x_k)$.
\end{definition}
Two point selectors are defined: the best selector, denoted by $\mathcal{S}_{+}$, selects the point with the best (i.e., least) pathloss, while the worst selector, denoted by $\mathcal{S}_{-}$,  selects the point with the worst (i.e, highest) pathloss.
For a given partition $\mathcal{B}$ with a given selection mode,  the best and the worst physically feasible networks are denoted as before  by $\mathfrak{N}_{+}(\mathcal{B})$ and  $\mathfrak{N}_{-}(\mathcal{B})$.

\begin{definition}[Relative transmission technique]
A transmission technique $(M_1,\dots,M_K,n,\epsilon)$ relative to a network $\mathfrak{N}(\mathcal{B})$ in MAC mode,  is given by a set of individual messages $M_k$ in $n$ channel uses, sent independently by the physically feasible transmitters, and where the unique receiver observes and decodes jointly the messages, denoted by $(\hat{m}_1, \dots, \hat{m}_K)$, with an average error probability lower than $\epsilon$. A transmission technique is asymptotically feasible without error if $\lim_{n\rightarrow\infty} \epsilon = 0$.
\end{definition}
Under these assumptions and for a given selector, $\mathfrak{N}(\mathcal{B})$ becomes equivalent to a classical Gaussian-MAC for which the capacity region is perfectly known.

\subsection{Asymptotic achievability, converse and capacity}
The definitions \ref{def:relative_achievability},\ref{def:asymptachievability} and \ref{def:accesscapacity} are general enough to apply for the SCMAC. 

Considering that the transmission powers are transferable between the different transmitters, the duality between MAC and BC is strict and theorems \ref{theo:achievability} and \ref{theo:converse} also apply where the best and worst receivers are resp. replaced by the best and wort transmitters. 
Last but not least, the uniform capacity and the capacity region are the same for the SCMAC and SCBC.
It is worth noting that even if the uniform capacity is the same, the corresponding optimal power allocation per node is not the same in both directions.

\section{Application example}
Consider a unique cell covering a disk (of radius $R$) in both uplink and downlink modes. 
For the sake of simplicity, a simple power-law pathloss  model and an omnidirectional antenna are considered with no shadowing:
$h(x)=h_0\cdot |x|^{-\alpha}$, 
where $h_0$ and $\alpha$ represent resp. the reference pathloss and the attenuation slope. 
Note that interference from neighboring cells is not considered. The reader is referred to \cite{gorce2014energy} where the idea of the SCBC has been implicitly used with an SINR distribution taking  interference into account.
The users are uniformly distributed, i.e.  $u(x)=u_0$ with constant quantity of information  $\mathcal{I}_0$.

\subsection{Downlink mode}
In the downlink,  the power transmission is limited to a certain power $P_M$. Using the model assumptions described above to solve \eqref{eq:ccdf}, the equivalent noise distribution is given by:
\begin{equation}
f_\nu(\nu) =  \frac{2}{\alpha} \cdot \left(\frac{\nu }{\nu_R}\right)^{2/\alpha-1} ,
\end{equation}
and the ccdf is denoted by $G(f)=1-  \left(\frac{\nu }{\nu_R}\right)^{2/\alpha}$,
where $\nu_R$ is the equivalent noise level at the cell edge. 

Under these assumptions, the following holds:
\begin{theorem}[Uniform capacity of an homogenous circular cell]
The uniform capacity of a wireless homogeneous cell with a power-law pathloss and with a radius $R$ is given by:
\begin{equation}
\label{eq:capaU_disk}
\mathcal{I}_0 \leq \mathcal{C}_0=\frac{1}{2\log2\cdot U_T}\cdot C_{U,\alpha}\left(\gamma_R \right) ,
\end{equation} 
where $\gamma_R$ is the SNR at the cell edge and $C_{U,\alpha}(\cdot)$ is the inverse function of $f(x)=x^{2+\frac{\alpha}{2}}\cdot {}_1\!F_1 \left(1;2+\frac{\alpha}{2} ; x \right)$ with ${}_1\!F_1(a;b;x)$ the confluent hypergeometric function (sec.9.21; \cite{gradshteyn2000table}). Note that $f(x)$ is continuous and strictly increasing and its inverse is defined and unique on $\mathbb{R}^{+}$.
\end{theorem}
\begin{proof}
The proof follows straightforward computations from Theorem\ref{theo:SCBC-unifcapa}, using the former assumptions, leading to the following: 
\begin{equation}
P_M=\nu(R)\cdot 2\eta_s\cdot e^{2 \eta_s}\cdot \Gamma\left(1+\frac{\alpha}{2}, 2 \eta_s \right).
\end{equation}
The final result is obtained using the relationship between the incomplete gamma function $\Gamma(a,x)$ and $\Phi(a,b;x)$. 
\end{proof}

It is remarkable that the uniform capacity relies only on the channel power law, the total number of users, and the SNR ($P_M/\nu(R)$) at the cell edge. This capacity is shown in Fig.\ref{fig:third}.

\begin{figure}[!t]
\centering
\includegraphics[width=5.5in]{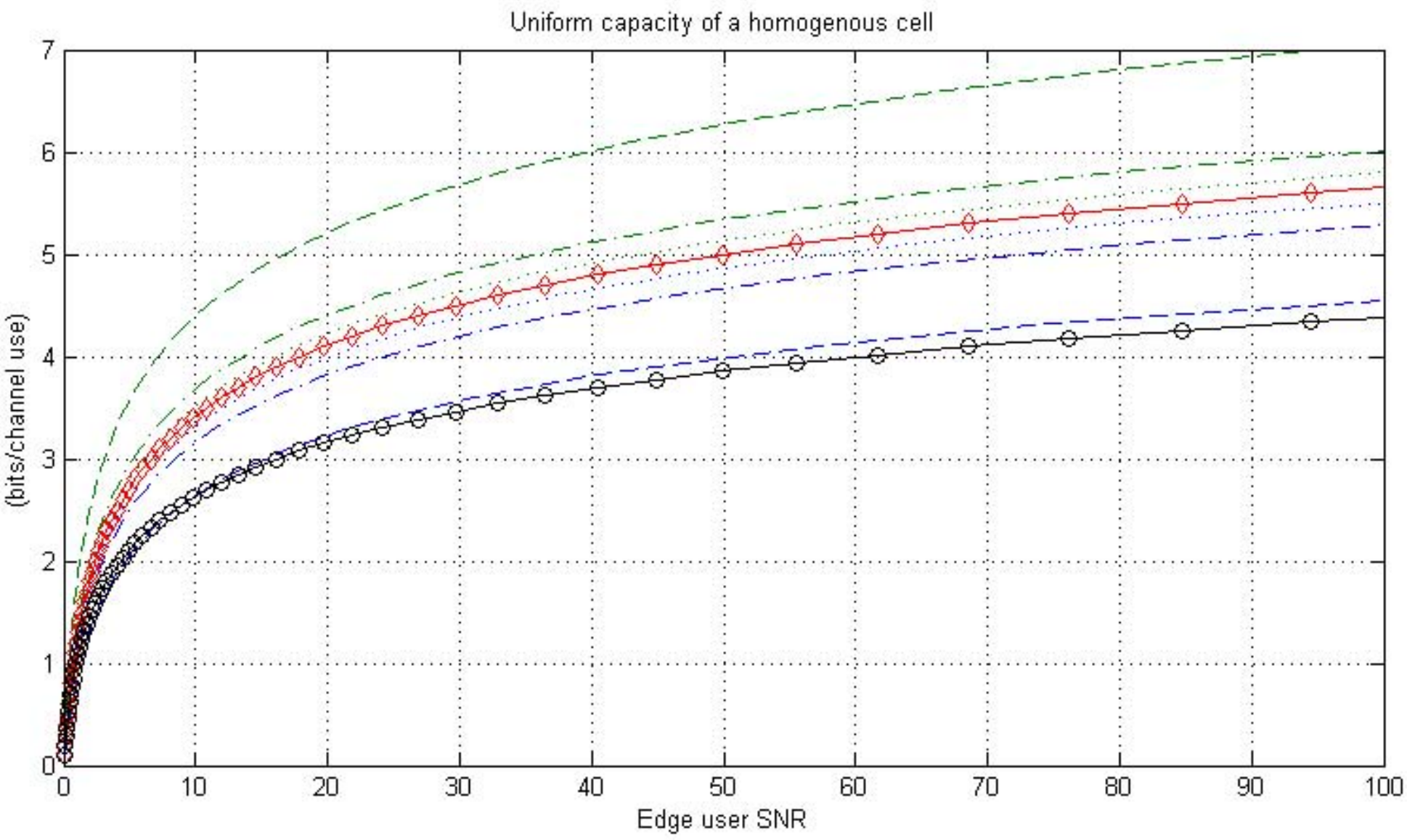}%
\caption{
Uniform capacity for a cell with $\alpha=3.65$ as a function of the edge user SNR. It is worth noting that this SNR is the ratio of the total power received by an edge user from the BS, to the receiver noise. This should not be confused with the effective SINR that is the ratio of the power devoted to this user to the noise plus the remaining power seen as interference. The result of Th.$5.1$ is represented in red (diamonds) under a normalized form ($\mathcal{I}_0\cdot U_T$), and the successive approximations with either best (green) or worst (blue) receivers, under a partition of $3$,$10$, and $25$ subsets (resp. dash, dash-dot and dot lines) are drawn. The uniform capacity theoretically achievable with pure time sharing is also plotted (black, circles).}
\label{fig:third}
\end{figure}
From a practical perspective, the theoretical limit relying on  superposition coding with an infinite number of layers is not feasible. The limit can be approached by grouping the users into $K$ subsets corresponding to a partition of the space as done for the proof of Theorem\ref{theo:SCBC-unifcapa}. The bound with the worst receivers shown in Fig.\ref{fig:third} is achievable with an appropriate scheduling and power association . It is only required to split the service area of a cell into $K$ subsets. Then, superposition coding between the different subsets can be used while inside each subset, a classical TDMA is performed. This approach highlights the potential gain of non orthogonal multiple access (NOMA) techniques for 5G.

\subsection{Uplink mode}
The SCMAC considered in this paper assumes transferable powers, and thus,  by the duality theorem, the uniform capacity is equal to that of the SCBC. 
From a practical perspective however, the power sharing rule is different.
The theoretical limit can be approached using one of the partitions $B^{(i)}$ of $\Omega$. For a given subset $B_k^{(i)}$, and to ensure the achievability of rates, the uniform capacity is lower bounded by using the worst transmitters, leading to a classical MAC with $K$ users. 
To allow perfect decoding, the power allocation starts with the furthest transmitter with a sufficient power to serve the rate of the subset $B_K$. Then, this power is considered as noise for the next power allocation:
\begin{equation}
P_k= \left( 2^{2\mathcal{R}_k^{(i)}} - 1\right) \cdot \left( 1 + \sum_{l>k} P_l / \nu_{M,l}^{(i)}  \right)\cdot \nu_{M,k}^{(i)}
\end{equation}
Under such an assumption, the uniform capacity can be approached, as well as for the SCBC. 
This provides a practical way to evaluate and exploit the additional capacity superposition coding may bring to cellular networks. In the setup above, the uniform capacity gain is about $30\%$ compared to time sharing.

\section{Conclusion}
In this paper we have proposed definitions of the Gaussian-SCBC and the Gaussian-SCMAC representing a wireless cell in downlink or uplink modes. We have defined the uniform capacity and the access capacity region and have given   general expressions for them by defining a sequence of partitions on $\Omega$, allowing us to approximate the continuum. 
With a simple example, we  have shown that the bound may be tight even for some discrete sets of users, and we have shown that the capacity thus defined represents a strict limit of a cell capacity.

The approach is limited to Gaussian channels, and further works will be done to exploit known results in information theory for the BC and MAC, to extend the ideas in this paper to fading channels, multi-antennas and multi-cells scenarios.



\bibliographystyle{IEEEtran}
\bibliography{IEEEabrv,ISIT-gorce}

\begin{thebibliography}{10}
\providecommand{\url}[1]{#1}
\csname url@samestyle\endcsname
\providecommand{\newblock}{\relax}
\providecommand{\bibinfo}[2]{#2}
\providecommand{\BIBentrySTDinterwordspacing}{\spaceskip=0pt\relax}
\providecommand{\BIBentryALTinterwordstretchfactor}{4}
\providecommand{\BIBentryALTinterwordspacing}{\spaceskip=\fontdimen2\font plus
\BIBentryALTinterwordstretchfactor\fontdimen3\font minus
  \fontdimen4\font\relax}
\providecommand{\BIBforeignlanguage}[2]{{%
\expandafter\ifx\csname l@#1\endcsname\relax
\typeout{** WARNING: IEEEtran.bst: No hyphenation pattern has been}%
\typeout{** loaded for the language `#1'. Using the pattern for}%
\typeout{** the default language instead.}%
\else
\language=\csname l@#1\endcsname
\fi
#2}}
\providecommand{\BIBdecl}{\relax}
\BIBdecl

\bibitem{el2011network}
A.~El~Gamal and Y.-H. Kim, \emph{Network information theory}.\hskip 1em plus
  0.5em minus 0.4em\relax Cambridge university press, 2011.

\bibitem{andrews2011tractable}
J.~G. Andrews, F.~Baccelli, and R.~K. Ganti, ``A tractable approach to coverage
  and rate in cellular networks,'' \emph{IEEE Transactions on Communications},
  vol.~59, no.~11, pp. 3122--3134, 2011.

\bibitem{elsawy2013stochastic}
H.~ElSawy, E.~Hossain, and M.~Haenggi, ``Stochastic geometry for modeling,
  analysis, and design of multi-tier and cognitive cellular wireless networks:
  A survey,'' \emph{IEEE Communications Surveys \& Tutorials}, vol.~15, no.~3,
  pp. 996--1019, 2013.

\bibitem{sang2014load}
Y.~J. Sang and K.~S. Kim, ``Load distribution in heterogeneous cellular
  networks,'' \emph{IEEE Communications Letters}, vol.~18, no.~2, pp. 237--240,
  2014.

\bibitem{bonald2003wireless}
T.~Bonald and A.~Prouti{\`e}re, ``Wireless downlink data channels: user
  performance and cell dimensioning,'' in \emph{Proc. 9th annual international
  conference on Mobile Computing and Networking (MOBICOM)}.\hskip 1em plus
  0.5em minus 0.4em\relax ACM, 2003, pp. 339--352.

\bibitem{minelli2014optimal}
M.~Minelli, M.~Ma, M.~Coupechoux, J.-M. Kelif, M.~Sigelle, and P.~Godlewski,
  ``Optimal relay placement in cellular networks,'' \emph{IEEE Transactions
  onWireless Communications}, vol.~13, no.~2, pp. 998--1009, 2014.

\bibitem{gorce2014energy}
J.-M. Gorce, D.~Tsilimantos, P.~Ferrand, and H.~V. Poor, ``Energy-capacity
  trade-off bounds in a downlink typical cell,'' in \emph{Proc. IEEE 25th
  International Symposium on Personal, Indoor and Mobile Radio Communications
  (PIMRC)}, 2014.

\bibitem{liang2006cth13}
Y.~Liang and A.~Goldsmith, ``Cth13-3: Symmetric rate capacity of cellular
  systems with cooperative base stations,'' in \emph{Proc. Global
  Telecommunications Conference (GLOBECOM). IEEE}.\hskip 1em plus 0.5em minus
  0.4em\relax IEEE, 2006, pp. 1--5.

\bibitem{jindal2004duality}
N.~Jindal, S.~Vishwanath, and A.~Goldsmith, ``On the duality of gaussian
  multiple-access and broadcast channels,'' \emph{IEEE Transactions on
  Information Theory}, vol.~50, no.~5, pp. 768--783, 2004.

\bibitem{cover1972broadcast}
T.~M. Cover, ``Broadcast channels,'' \emph{IEEE Transactions on Information
  Theory}, vol.~18, no.~1, pp. 2--14, 1972.

\bibitem{gradshteyn2000table}
I.~S. Gradshteyn and I.~Ryzhik, ``Table of integrals, series, and products.
  translated from the russian. translation edited and with a preface by alan
  jeffrey and daniel zwillinger,'' 2000.

\end{thebibliography}

\end{document}